\documentclass[a4paper, english, final]{article}

\usepackage[english]{babel}
\usepackage[applemac]{inputenc}
\usepackage{showkeys}

\usepackage{microtype}

\usepackage{pgfplots} 
\usepackage{amsfonts}
\usepackage{amssymb}
\usepackage{amsmath}
\usepackage[arrow, matrix, curve]{xy}
\usepackage{enumerate,xspace}
\usepackage{tikz}
\usetikzlibrary{trees}
\usepackage{ellipsis}
\usepackage{ifthen}
\usepackage{ stmaryrd }

\numberwithin{equation}{section}

\usepackage{theorem}
\usepackage{prettyref}

\newrefformat{thm}{Thm.~\ref{#1}}
\newrefformat{lm}{Lem.~\ref{#1}}
\newrefformat{dfn}{Def.~\ref{#1}}
\newrefformat{cor}{Cor.~\ref{#1}} 
\newrefformat{prop}{Prop.~\ref{#1}}
\newrefformat{sec}{Sect.~\ref{#1}}
\newrefformat{kap}{Chap.~\ref{#1}}
\newrefformat{bsp}{Ex.~\ref{#1}}
\newrefformat{app}{App.~\ref{#1}}
\newrefformat{alg}{Alg.~\ref{#1}}
\newrefformat{eq}{Eq.~(\ref{#1})}%
\newrefformat{tab}{Table~\ref{#1}}
\newrefformat{fig}{Fig.~\ref{#1}}

\newtheorem{theorem}{Theorem}[section]

\newtheorem{lemma}[theorem]{Lemma}
\newtheorem{proposition}[theorem]{{\bf Proposition}}
\newtheorem{corollary}[theorem]{Corollary}

\newenvironment{proof}[1][Proof.]{
\begin{trivlist}
\item[\hskip \labelsep {\bfseries #1}]}{\hspace*{\fill}$\Box$\end{trivlist}
}

\renewcommand{\Pr}[1]{\mathop{\mathrm{Pr}}\left[\,#1\,\right]}
\newcommand{\ET}[1]{\mathbb{E}\!\left[\,T(#1)\,\right]}
\newcommand{\Ex}[1]{\mathbb{E}\!\left[\,#1\,\right]}
\newcommand{\hei}{\mathop{\mathrm{height}}}
\newcommand{\lev}{\mathop{\mathrm{level}}}

\newcommand{\Ep}{E}

\newcommand{\Oh}{\mathcal{O}}

\newcommand{\abs}[1]{\left|\mathinner{#1}\right|}

\newcommand{\floor}[1]{\left\lfloor\mathinner{#1} \right\rfloor}
\newcommand{\ceil}[1]{\left\lceil\mathinner{#1} \right\rceil}

\newcommand{\os}[1]{\left\{\mathinner{#1}\right\}}

\newcommand{\sse}{\subseteq}

\newcommand{\N}{\mathbb{N}}
\newcommand{\Z}{\mathbb{Z}}

\newcommand{\R}{\mathbb{R}}

\renewcommand{\phi}{\varphi}
\newcommand{\eps}{\varepsilon}

\newcounter{AlgorithmusCounter}[section]
\renewcommand{\theAlgorithmusCounter}{\arabic{section}.\arabic{AlgorithmusCounter}}
\newenvironment{algorithmus}[1]{%
  \refstepcounter{AlgorithmusCounter}%
  \begin{tabbing}
    \hspace{9cm}\=\kill
    \tab\rule{\textwidth}{0.8pt} \\
    \parbox{.99\textwidth}{\textbf{Algorithm \theAlgorithmusCounter\ }
    \emph{#1}} \\[-2.5mm]
    \tab\rule{\textwidth}{0.4pt} \\
}{
    \\[-3mm]
    \tab\rule{\textwidth}{0.8pt}
  \end{tabbing}
}

\newcommand{\comment}[1]{ \> \tab $(*$ {\footnotesize #1} $*)$}  

\newcommand{\cofunction}{\textbf{function }}
\newcommand{\coprocedure}{\textbf{procedure }}
\newcommand{\coendfunction}{\textbf{endfunction }}
\newcommand{\coendprocedure}{\textbf{endprocedure }}

\newcommand{\cobegin}{\textbf{begin }}

\newcommand{\coif}{\textbf{if }}
\newcommand{\cothen}{\textbf{then }}
\newcommand{\coelse}{\textbf{else }}

\newcommand{\cofor}{\textbf{for }}

\newcommand{\codo}{\textbf{do }}
\newcommand{\coendif}{\textbf{endif }}
\newcommand{\coendfor}{\textbf{endfor }}
\newcommand{\coreturn}{\textbf{return }}

\newcommand{\coto}{\textbf{to }}

\newcommand{\cowhile}{\textbf{while }}
\newcommand{\coendwhile}{\textbf{endwhile }}

\newcommand{\coand}{\textbf{and }}

\newcommand{\IndentLength}{\hspace*{1em}}
\newcommand{\tab}{}
\newcommand{\tabb}{\IndentLength}
\newcommand{\tabbb}{\tabb\IndentLength}
\newcommand{\tabbbb}{\tabbb\IndentLength}

\renewcommand{\log}{\lg}

\title{QuickHeapsort: Modifications  and Improved Analysis}

\author{Volker Diekert \and Armin Wei\ss}

\begin{document}

\author{Volker Diekert \qquad Armin Wei\ss \\[5mm]
 Universit{\"a}t Stuttgart, FMI \\
 Universit{\"a}tsstra{\ss}e 38 \\
 D-70569 Stuttgart, Germany \\[5mm]
 \texttt{$\{$diekert$,$weiss$\}$@fmi.uni-stuttgart.de}}

\maketitle

\begin{abstract}
  \noindent
We present a new analysis for QuickHeapsort splitting it into the analysis of the partition-phases and the analysis of the heap-phases. This enables us to consider samples of non-constant size for the pivot selection and leads to better theoretical bounds for the algorithm. 

Furthermore we introduce some modifications of QuickHeapsort, both in-place and using $n$ extra bits. We show that on every input the expected number of comparisons is $n\lg n - 0.03n + o(n)$  (in-place) respectively $ n\lg n -0.997 n+ o (n)$ (always $\lg n= \log_2 n$).  Both estimates improve the previously known best results. (It is conjectured \cite{Wegener93} that the in-place algorithm Bottom-Up-Heapsort uses at most $n\lg n + 0.4 n$ on average and for Weak-Heapsort which uses $n$ extra bits the average number of comparisons is at most $n\lg n -0.42n$ \cite{EdelkampS02}.) Moreover, our non-in-place variant can even compete with index based Heapsort variants  (e.g.\ Rank-Heapsort \cite{WangW07}) and Relaxed-Weak-Heapsort ($ n\lg n -0.9 n+ o (n)$ comparisons in the worst case) for which no $\Oh(n)$-bound on the number of extra bits is known.
  \medskip

\end{abstract}
\section{Introduction}
QuickHeapsort is a combination of Quicksort and Heapsort which was first described by Cantone and Cincotti \cite{CantoneC02}. It is based on Katajainen's idea for  Ultimate Heapsort \cite{Katajainen98}. In contrast to Ultimate Heapsort it does not have any  $\mathcal{O}(n\lg n)$ bound for the worst case running time ($\lg n= \log_2 n$). Its advantage is that it is very fast in the average case and hence not only of theoretical interest.

Both algorithms have in common that first the array is partitioned into two parts. Then in one part a heap is constructed and the elements are successively extracted. Finally the remaining elements are treated recursively. The main advantage of this method is that 
 for the sift-down only <one comparison per level is needed, whereas standard Heapsort needs two comparisons per level (for a description of standard Heapsort see some standard textbook, e.g.\ \cite{CLRS09}). This is a severe drawback and one of the reasons why standard Heapsort cannot compete with Quicksort in practice (of course there are also other reasons like cache behavior). Over the time a lot of  solutions to this problem appeared like Bottom-Up-Heapsort \cite{Wegener93} or MDR-Heapsort \cite{McDiarmidR89},\cite{Wegener91}, which both perform the sift-down by first going down to some leaf and then searching upward for the correct position. Since one can expect that the final position of some introduced element is near to some leaf, this is a good heuristic and it leads to provably good results.
The difference between QuickHeapsort and Ultimate Heapsort lies in the choice of the pivot element for partitioning the array. While for Ultimate Heapsort the pivot is chosen as median of the whole array, for QuickHeapsort the pivot is selected as median of some smaller sample (e.g.\ as median of 3 elements).

  In \cite{CantoneC02} the basic version with fixed index as pivot is analyzed and~-- together with the median of three version~-- implemented and compared with other Quick- and Heapsort variants. In \cite{EdelkampS02} Edelkamp and Stiegeler compare these variants with so called Weak-Heapsort \cite{Dutton93} and some modifications of it (e.g.\ Relaxed-Weak-Heapsort). 
 Weak-Heapsort beats basic QuickHeapsort with respect to the number of comparisons, however it needs $\mathcal{O}(n)$ bits extra-space (for Relaxed-Weak-Heapsort this bound is only conjectured), hence is not in place.

We split the analysis of QuickHeapsort into three parts: the partitioning phases, the heap construction and the heap extraction. This allows us to get better bounds for the running time, especially when choosing the pivot as median of a larger sample. It also simplifies the analysis.
We introduce some modifications of QuickHeapsort, too. The first one is in-place and needs $n\lg n - 0.03n + o(n)$ comparisons on average what is to the best of our knowledge better than any other known in-place Heapsort variant. 
We also examine a modification using $\mathcal{O}(n)$ bits extra-space, which applies the ideas of MDR-Heapsort to QuickHeapsort. With this method we can bound the average number of comparisons to $ n\lg n -0.997 n+ o (n)$.
Actually, a complicated, 
iterated in-place MergeInsertion uses only $n\lg n -1.3n+\mathcal{O}(\lg n)$ comparisons, \cite{Reinhardt92}. Unfortunately, for practical purposes this algorithm is 
not competitive.

Our contributions are as follows: 
1. We give a simplified analysis which gives better bounds than previously known. 
2. Our approach yields the first precise analysis of  QuickHeapsort when the pivot element is taken from a 
larger sample. 
3. We give a simple in-place modification of QuickHeapsort which saves $0.75 n$ comparisons. 
4. We give a modification of QuickHeapsort using $n$ extra bits only and 
we can bound the expected number of comparisons. This bound is 
better than  the  previously known for the worst case of Heapsort variants using $\Oh(n\lg n)$ extra bits for which best and worst case are almost the same. 
5. We have implemented QuickHeapsort, and our experiments confirm the theoretical predictions.

The paper is organized as follows: \prettyref{sec:QuickHeapsort}  briefly describes the basic QuickHeapsort algorithm together 
with our first improvement. 
In \prettyref{sec:analysis} we analyze the expected running time of QuickHeapsort.
Then we introduce some improvements in \prettyref{sec:modifications} allowing $\Oh(n)$ additional bits. 
Finally, in  \prettyref{sec:experiments}, we present our experimental results comparing the different versions of QuickHeapsort with other Quicksort and Heapsort variants.

\section{QuickHeapsort}\label{sec:QuickHeapsort}

A \emph{two-layer-min-heap} is an array $A[1..n]$ of $n$ elements together with 
 a partition $(G,R)$ of $\{1,\dots,n\}$ into \emph{green} and \emph{red} elements such that for all $g\in G, r\in R$ we have $A[g] \leq A[r]$. Furthermore, the green elements $g$ satisfy the heap condition $A[g] \leq\min\{ A[2g],A[2g +1]\}$, and if  $g$ is red, then $2g$ and $2g +1$ are red, too. 
 (The conditions are required to hold, only if the indices involved are in the range of 
 $1$ to $n$.) 
The green elements are called ``green'' because the they can be extracted out of the heap without caution, whereas the ``red'' elements are blocked. \emph{Two-layer-max-heaps} are defined analogously. We can think 
of a two-layer-heap as rooted binary tree such that each node is either green or red. 
Green nodes satisfy the standard heap-condition, children of red nodes are red.
Two-layer-heaps were defined in \cite{Katajainen98}. In \cite{CantoneC02} for the same concept a different language is used (they describe the algorithm in terms of External Heapsort). 
 Now we are ready to describe the QuickHeapsort algorithm
 as it has been proposed in \cite{CantoneC02}. Most of it also can be found in pseudocode in \prettyref{app:pseudo_basic}.

 We intend to sort  an array $A[1.. n]$. 
 First, we choose a \emph{pivot} $p$. This is the randomized part of the algorithm. Then, just as in Quicksort,
 we rearrange the array according to $p$.  That means, 
 using $n-1$ comparisons the partitioning function returns an index $k$ and rearranges the array $A$ so that  
 $A[i]\geq  A[k]$ for $i< k$, $A[k] =p$, and $A[k]\geq A[j]$ for $k< j$. 
 After the partitioning a two-layer-heap is built out of the elements 
  of the smaller part of the array, either the part left of the pivot or right of the pivot. We call this smaller part \emph{heap-area} and the larger part \emph{work-area}. 
  More precisely, if $k-1 < n-k$, then $\os{1, \ldots, k-1}$ is the  heap-area and
  $\os{k+1, \ldots, n}$ is the  work-area. If $k-1 \geq  n-k$, then $\os{1, \ldots, k-1}$ is the  work-area and
  $\os{k+1, \ldots, n}$ is the  heap-area. Note that we know the final position of the pivot element without any further comparison. 
Therefore, we do not count it to the heap-area nor to the work-area.
  If the heap-area the part of the array left of the pivot, a two-layer-max-heap is built, otherwise a two-layer-min-heap is built.

  At the beginning the heap-area is an ordinary heap, hence it is a
two-layer-heap consisting  of green elements, only. 
  Now the heap extraction phase starts. We assume that we are in the case of a max-heap. The other case is symmetric. Let $m$ denote the size of the heap-area. 
The $m$ elements of the heap-area are moved to the work-area. 
The extraction of one element works as follows: the root of the heap is placed at the current position of the work-area (which at the beginning is its last position). Then, starting from the root the resulting \lq\lq{}hole\rq\rq{} is trickled down: always the larger child is moved up into the vacant position and then this child is treated recursively. This stops as soon as a leaf is reached. We call this the SpecialLeaf procedure (\prettyref{alg:special_leaf}) according to \cite{CantoneC02}. 
 Now, the element which before was at the current position in the work-area is placed as red element in this hole at the leaf in the heap-area.
Finally the current position in the work-area is moved by one and the next element can be extracted.

The procedure sorts correctly, because after the partitioning it is guaranteed that all red elements are smaller than all green elements. Furthermore there is enough space in the work-area to place all green elements of the heap, since the heap is always the smaller part of the array.
After extracting all green elements the pivot element it placed at its final position and the remaining elements are sorted recursively.

Actually we can improve the procedure, thereby saving  $3n/4$ comparisons by a simple trick. 
Before the heap extraction phase starts in the heap-area with $m$ elements, we perform at most $\frac{m+2}{4}$ additional comparisons in order to  arrange all pairs of leaves which share a parent such that the left child is 
not smaller than its right sibling. 
 Now, in every call of SpecialLeaf, we can save exactly one comparison, since we do not need to compare two leaves. For a max-heap we only need to move up the left child and put the right one at the place of the former left one. Summing up over all heaps during an execution of standard QuickHeapsort, we invest $\frac{n+2t}{4}$ comparisons in order to save $n$ comparisons, where 
 $t$ is the number of recursive calls. The expected number of $t$ is in $\Oh(\lg n)$.  
 Hence, we can expect to save $\frac{3n}{4}+\Oh(\lg n)$ comparisons.
 We call this version the \emph{improved} variant of  QuickHeapsort.\\

\section{Analysis of QuickHeapsort}\label{sec:analysis}
This section contains the main contribution of the paper. 
We analyze the number of comparisons. 
By $n$ we  denote the number of elements of an array  to be sorted.
We use standard $\Oh$-notation where $\Oh(g)$, $o(g)$, and $\omega(g)$  denote classes of functions. 
In our analysis we do not assume any random distribution of the input, i.e. it is valid for every permutation of the input array. 
  Randomization is used however for pivot selection. 
With $\Pr{e}$ we denote the probability of some event $e$. The expected 
value of a random variable $T$ is denoted by $\Ex{T}$.
 
The number of assignments is bounded by some small constant times the number of comparisons. Let $T(n)$ denote 
the number of comparisons during QuickHeapsort on a fixed array of $n$ elements.
 We are going to split the analysis of QuickHeapsort into three parts:
\begin{enumerate}
\item Partitioning with an expected number of comparisons $\Ex{T_{\mathrm{part}}(n)}$ (average case).
\item Heap construction with at most $T_{\mathrm{con}}(n)$ comparisons (worst case).
\item Heap extraction (sorting phase) with at most $T_{\mathrm{ext}}(n)$
comparisons (worst case).
\end{enumerate}
We analyze the three parts separately and put them together at the end. 
The partitioning is the only randomized part of our algorithm. 
The expected number of comparisons depends on the selection method for the pivot.
For the expected number of comparisons
by QuickHeapsort on the input array we obtain $\ET{n} \leq T_{\mathrm{con}}(n)
+T_{\mathrm{ext}}(n) +{\mathbb{E}}[T_{\mathrm{part}}(n)]$.
\begin{theorem}\label{thm:exptime}
The expected number  $\ET{n}$ of comparisons
by basic resp.\ improved QuickHeapsort with pivot as median of $p$ randomly selected elements on a fixed input array of size $n$ is $\ET{n}\leq n\lg n + c n+ o(n)$ with $c$ as follows:
\begin{center}
\begin{tabular}{c|c|c}
$p$  ~&~$c$ basic~ &~$c$ improved\\\hline
$1$ &  $+2.72$ & $+1.97$ \\
$3$ &   $+1.92$ & $+1.17$ \\
$f(n)$ & $+0.72 $ & $-0.03$\\
\end{tabular}
\end{center}
Here, $f\in \omega(1)\cap o(n)$ with $1\leq f(n) \leq n$, e.g., 
$f(n) = \sqrt{n}$ and we assume that we choose the median of $f(n)$ randomly selected elements 
in time $\Oh(f(n))$.
\end{theorem}

As we see, the selection method for the pivot is very important. However, one should notice that the bound for fixed size samples for pivot selection are not tight. 
The proof of these results are postponed to \prettyref{sec:pr}.
Note that it is enough to prove the results without the improvement, 
since the difference is always $0.75n$. 

\subsection{Heap Construction}\label{sec:heap_con}
The standard heap construction \cite{Floyd64} needs at most $2m$ comparisons to construct a heap of size $m$ in the worst case and approximately $1.88 m$ in the average case. 
For the mathematical analysis better theoretical bounds can be used. 
The best result we are aware of is due to Chen et al. in \cite{ChenEEK12}. According to this result we have
$T_{\mathrm{con}}(m)\leq 1.625 m + o(m)$.
Earlier results are of similar magnitude, by \cite{Chen93} it has been known that $T_{\mathrm{con}}(m)\leq 1.632 m + o(m)$ and by \cite{GonnetM86} it has been
known $T_{\mathrm{con}}(m)\leq 1.625 m + o(m)$, but Gonnet and Munro used $\Oh(m)$
extra bits to get this result, whereas the new result of Chen et al.~is in-place
(by using only $\Oh(\lg m)$ extra bits).

During the execution of QuickHeapsort over $n$ elements, every element is part of a heap only once.  Hence, the sizes of all heaps during the entire procedure sum up to $n$. With the result of \cite{ChenEEK12}
 the total number of comparisons performed in the construction of all heaps satisfies:
\begin{proposition}\label{prop:chen}
$ T_{\mathrm{con}}(n) \leq  1.625 n + o(n)$. 
\end{proposition}

\subsection{Heap Extraction}\label{sec:oje}
For a real number $r\in \R$ with $r> 0$ we define $\os{ r}$ by the following condition
\[r = 2^k + \os r \text{ with $k \in \Z$ and $0\leq \os r < 2^k$} . \]
This means that $2^k$ is largest power of $2$ which is less than or equal to $r$
and $\os r$ is the difference to that power, i.e. $\os r = r - 2^{\floor{\lg r}}$. 
In this section we first analyze the extraction phase of one two-layer-heap of size $m$. After that, we bound the number of comparisons $T_{\mathrm{ext}}(n)$ performed in the worst case during all heap extraction phases of one execution of QuickHeapsort on an array of size $n$. 
\prettyref{thm:heapext} is our central result about heap extraction. 

\begin{theorem}\label{thm:heapext}
$T_{\mathrm{ext}}(n) 	\leq n \cdot\left(\lfloor\lg n\rfloor-3\right) + 2\{ n\}+ \mathcal{O}(\lg^2 n)  .$
\end{theorem}

\newcommand{\pp}{v}

The proof of \prettyref{thm:heapext} covers  almost the rest of Section~\ref{sec:oje}. 
In the following, the \emph{height} $\hei(\pp)$ of an element $\pp$ in a heap $H$ is the maximal distance from that node to a leaf below it. The \emph{height} of $H$ is the height 
of its root. 
The \emph{level} $\lev(\pp)$ of $\pp$  to be its distance from the root.
In this section we want to count the comparisons during SpecialLeaf procedures, only. 
Recall that a SpecialLeaf procedure is a cyclic shift on a path from the root down to
some leaf, and the number comparisons is exactly the length of this path. Hence the upper bound is the height of the heap. But there is a better analysis. 

Let us consider a heap with $m$ green elements which are all extracted 
by SpecialLeaf procedures. The picture is as follows: First, we color the green root red. Next, we perform a cyclic shift defined by the SpecialLeaf procedure. In particular, the leaf is now red. Moreover, red positions remain red, but there is exactly one position $\pp$ which has changed its color from green to red. This position
$\pp$ is on the path defined  by the SpecialLeaf procedure. Hence, the number of comparisons needed to color the position $\pp$ red is bounded by $\hei(\pp) + \lev(\pp)$.

The total number of comparisons $E(m)$ to extract all $m$ elements of a Heap $H$ is therefore bounded by
$$E(m)  \leq \sum_{\pp \in H} (\hei(\pp) + \lev(\pp)).$$

We have $\hei (H)- 1 \leq \hei(\pp) + \lev(\pp) \leq  \hei (H)= \floor{\lg m} $ for all $\pp \in H$. 
We now count the number of elements $\pp$ where $\hei(\pp) + \lev(\pp) = \floor{\lg m}$ and the number of elements $\pp$ where $\hei(\pp) + \lev(\pp) = \floor{\lg m}-1$.
Since there are exactly $\os m +1 $ nodes of level $\floor{\lg m}$,
 there are at most $2\os m + 1 + \lg m$ elements $\pp$ with $\hei(\pp) + \lev(\pp) = \floor{\lg m}$. All other elements satisfy $\hei(\pp) + \lev(\pp) = \floor{\lg m}-1$. 
 We obtain 
 \begin{align}
E(m) 	&\leq 2\cdot\{m\}\cdot \lfloor\lg m\rfloor + (m - 2\cdot\{m\}) ( \lfloor\lg m\rfloor -1) + \mathcal{O}(\lg m)\nonumber\\
		&= m\cdot (\floor{\lg m}-1)+ 2\cdot\{m\}  + \mathcal{O}(\lg m) . \label{eq:E_m}
\end{align}
Note that this is an estimate of the worst case, however this analysis also shows that the best case only differs by $\Oh(\lg m)$-terms from the worst case.\\

Now, we want to estimate the number of comparisons in the worst case performed during  all heap extraction phases together. 
During QuickHeapsort over $n$ elements we create  a sequence $H_1, \ldots , H_t$ 
of heaps of green elements which are extracted using the SpecialLeaf procedure. Let $m_i = \abs{H_i}$ be the size of the $i$-th Heap. 
The sequence satisfies 
$2m_i \leq n-\sum_{j<i} m_j$, because heaps are constructed and extracted on  the smaller part of the array.

Here comes a subtle observation: Assume that $m_1 + m_2\leq n/2$. If we replace 
the first two heaps with one heap $H'$ of size $\abs H' = m_1 +m_2$, then 
the analysis using the sequence $H', H_3, \ldots , H_t$ cannot lead to a better bound. Continuing this way, we may assume
that we have $t \in \Oh(\lg n)$ and therefore 
$\sum_{1 \leq i \leq t} \mathcal{O}(\lg m_i) \sse \mathcal{O}(\lg^2 n).$
With \prettyref{eq:E_m} we obtain the bound
\begin{align}
T_{\mathrm{ext}}(n) 	\leq \sum_{i= 1}^{t} E(m_i) &= 
\left(\sum_{i= 1}^t  m_i\cdot \floor{\lg m_i} +2 \os{m_i}\right)  -n + \mathcal{O}(\lg^2 n) .\label{eq:T_ext_ineq}
\end{align}

Later we will replace the $m_i$ by other positive real numbers. Therefore we define the following notion. Let $1 \leq \nu \in \R$. We say a sequence $x_1,x_2,\ldots, x_t$ with $x_i\in \R^{>0}$ is \emph{valid} w.r.t.{} $\nu$, if for all $1\leq i \leq t$ we have 
$2x_i \leq \nu-\sum\limits_{j<i} x_j.$

As just mentioned the initial sequence $m_1,m_2\ldots , m_t$ is valid 
w.r.t.{} $n$. 
Let us define a continuous function $F:  \R^{>0} \to \R$ by
$F(x) = x\cdot \floor{\lg x} +2 \os{x}.$
It  is continuous since for $x = 2^k$, $k \in \Z$ we have
$F(x) = xk = \lim_{\eps \to 0} (x -\eps)(k-1) +2\os{x-\eps}$. 
It is piecewise differentiable with right derivative 
$\floor{\lg x} + 2$. Therefore: 
\begin{lemma}\label{lm:F_ineq}
Let $x \geq y > \delta \geq 0$. Then we have the inequalities:
 \[F(x) + F(y) \leq F(x+\delta) + F(y - \delta) \text{ and } F(x) + F(y) \leq F(x+y).\]
\end{lemma}

\begin{lemma}\label{lm:fred}
Let $1 \leq \nu \in \R$. 
For  all sequences  $x_1,x_2,\dots, x_t$ with $x_i\in \R^{>0}$, which are valid w.r.t.{} $\nu$, we have
$\sum\limits_{i= 1}^{t} F(x_i) \leq \sum\limits_{i= 1}^{\floor{\lg \nu}} F\left(\frac{\nu}{2^i}\right)$.
\end{lemma}

\begin{proof}
The result is true for $\nu \leq 2$, because then $F(x_i) \leq F(\nu / 2)
\leq F(1) = 0$ for all $i$. Thus, we may assume $\nu \geq 2$. 
We perform induction on $t$. For $t=1$ the statement is clear, since $\lg \nu \geq 1$ and $x_1 \leq \nu/ 2$. Now let $t>1$.  By \prettyref{lm:F_ineq}, we have $F(x_1) + F(x_2) < F(x_1 + x_2)$. Now, if $x_1 + x_2 \leq \frac{\nu}{2}$, then the  sequence $x_1 + x_2,x_3, \dots, x_t$ is valid, too; and we are done by induction. Hence, we may assume  $x_1 + x_2 > \frac{\nu}{2}$. 
If $x_1\leq x_2$, then 
\begin{align*}
2x_1 &= 2x_2 + 2(x_1 - x_2) \leq \nu-x_1 + 2(x_1 - x_2)
=\nu-x_2 + x_1 - x_2 \leq \nu-x_2.
\end{align*}
Thus, if $x_1\leq x_2$, then the sequence $x_2,x_1,x_3,\ldots, x_t$ is valid, too. Thus, it is enough to consider $x_1\geq x_2$ with $x_1 + x_2 > \frac{\nu}{2}$.

%

We have  $\frac{\nu}{2} \geq 1$ and 
the sequence $x'_2,x_3,\dots x_t$ with $x'_2= x_1+x_2-\frac{\nu}{2}$ is valid w.r.t.{} $\nu/2$, because
\[x'_2=x_1+x_2-\frac{\nu}{2}\leq x_1 +\frac{\nu-x_1}{2} -\frac{\nu}{2}= \frac{x_1}{2}\leq \frac{\nu}{4} .\]  Therefore, by induction on $t$ and \prettyref{lm:F_ineq} we obtain the claim: 
\begin{equation*}
\sum_{i= 1}^{t} F(x_i) \leq F(\nu/2)+ F(x'_2)+\sum_{i= 3}^{t} F(x_i) 
\leq F(\nu/2)+  \sum_{i= 2}^{\floor{\lg \nu}} F\left(\frac{\nu}{2^i}\right)
\leq \sum_{i= 1}^{\floor{\lg \nu}} F\left(\frac{\nu}{2^i}\right) .\end{equation*}
%
\end{proof}

\begin{lemma}\label{lm:F_sum}
$ \sum\limits_{i= 1}^{\floor{\lg n}} F\left(\frac{n}{2^i}\right) 
\leq F(n) -2n + \Oh(\lg n)$.
\end{lemma}

\begin{proof}[\prettyref{lm:F_sum}]
\begin{align*}  
 \sum_{i= 1}^{\floor{\lg n}} F\left(\frac{n}{2^i}\right) 
&=  n\floor{\lg n} \cdot\sum_{i= 1}^{\floor{\lg n}}\frac{1}{2^{i}}- n \cdot \sum_{i= 1}^{\floor{\lg n}}\frac{i}{2^{i}} + 2\os n  \cdot\sum_{i= 1}^{\floor{\lg n}}\frac{1}{2^{i}}\\
&\leq  n\floor{\lg n} \cdot\sum_{i\geq 1}\frac{1}{2^{i}}- n \cdot \sum_{i\geq 1}\frac{i}{2^{i}} + 2\os n \cdot \sum_{i\geq 1}\frac{1}{2^{i}} +\frac{n}{2^{\floor{\lg n}}} \cdot \sum_{i > 0} \frac{i+\floor{\lg n}}{2^{i}} \\
&= n {\floor{\lg n}} -2n + 2 \os{n}  + \Oh(\lg n) .
\end{align*}
\end{proof}

%
%
%
Applying these lemmata to \prettyref{eq:T_ext_ineq} yields the proof of \prettyref{thm:heapext}.

\begin{corollary}\label{cor:T_ext}
We have
$T_{\mathrm{ext}}(n) 	\leq  n\lg n - 2.9139 n + \mathcal{O}(\lg^2 n).$
\end{corollary}
\begin{proof}
By \cite[Thm. 1]{Wegener91} we have $F(n)-2n \leq  n\lg n - 1.9139 n $. Hence, \prettyref{cor:T_ext} follows directly from \prettyref{thm:heapext}.
\end{proof}

\subsection{Partitioning}\label{sec:pr}
In the following $T_{\mathrm{pivot}}(n)$ denotes the number of comparisons required to choose the pivot element in the worst case; and, as before,  $\Ex{T_{\mathrm{part}}(n)}$  denotes  the expected  
number of comparisons  performed during partitioning.  We have  the following recurrence:
\begin{align}\label{eq:gunnar}
 \Ex{T_{\mathrm{part}}(n)} &\leq n - 1 + T_{\mathrm{pivot}}(n) + \sum_{k=1}^n \Pr{\!\text{pivot }= k\!}\cdot \Ex{\!T_{\mathrm{part}}(\max\os{k-1,n-k})\!} .
\end{align}
If we choose the pivot at random, then we obtain  by standard methods:
\begin{align}\label{gunone}
 \Ex{T_{\mathrm{part}}(n)}  \leq n-1 + \frac{1}{n}\cdot \sum_{k=1}^n \Ex{T_{\mathrm{part}}(\max\os{k-1,n-k})}\leq 4n .
 \end{align}
Similarly, if we choose the pivot with the  median-of-three, then we obtain:
\begin{align}\label{gunthree} \Ex{T_{\mathrm{part}}(n)}  \leq 3.2n + \mathcal{O}(\lg n) .
\end{align}

The proof of the first part of \prettyref{thm:exptime} follows from the above eqations, \prettyref{thm:heapext}, 
and \prettyref{prop:chen}.
Using a growing number of elements  (as $n$ grows) as sample for the pivot selection,  we can do better. The second part of \prettyref{thm:exptime} follows from  \prettyref{thm:heapext}, \prettyref{prop:chen}, and \prettyref{thm:otto}.

\begin{theorem}\label{thm:otto}
Let $f\in \omega(1)\cap o(n)$ with $1\leq f(n) \leq n$.
When choosing the pivot as median of $f(n)$ randomly selected elements 
in time $\Oh(f(n))$ (e.g.\ with the algorithm of \cite{BlumFPRT73}), 
 the expected number of comparisons used in all recursive calls of partitioning is in $2n + o(n)$.
\end{theorem}

\prettyref{thm:otto} is close to a well-known result in \cite[Thm. 5]{MartinezR01}
on Quickselect, see  \prettyref{cor:quickselect}. Formally speaking we cannot use it directly, 
because we deal with QuickHeapsort, where after partitioning the recursive call is on the larger part. 
Because of that, and for the sake of completeness, we give a proof. 
Moreover, our proof is elementary and simpler than the one in \cite{MartinezR01}. The key step is \prettyref{lm:prob_bound}. Its proof is rather standard and also can be found in 
\prettyref{app:proof}. 

\begin{lemma}\label{lm:prob_bound}
Let $0<\delta < \frac{1}{2}$. If we choose the pivot as median of $2c +1$ elements such that $2c +1 \leq\frac{n}{2}$, then we have $\Pr{\text{pivot }\leq \frac{n}{2} - \delta n} < (2c+1) \alpha^c$ where $\alpha = 4\left(\frac{1}{4} - \delta^2\right)<1$.
\end{lemma}

\begin{proof}[Proof of \prettyref{thm:otto}.]
As an abbreviation, we let
 $\Ep(n) = \Ex{T_{\mathrm{part}}(n)}$  be the expected  
number of comparisons  performed during partitioning.
We are going to show that for all $\epsilon >0$ there is some $D\in \R$ such that
\begin{align}
E(n) < (2+\epsilon)n + D .
\label{eq:to_show}
\end{align}
So, we fix some $1 \geq \epsilon >0$. 
We choose $\delta > 0$ 
such that $(2+\epsilon) \delta< \frac{\epsilon}{4}.$
Moreover, for this proof let  $\mu = \frac{n+1}{2}$. 
Positions of possible pivots $k$ with 
$\mu - \delta n \leq k \leq \mu + \delta n$ form a small fraction 
of all positions, and they are located around the 
median. Nevertheless, 
applying \prettyref{lm:prob_bound} with $c=f(n)\in \omega(1)\cap o(n)$ yields 
  for all $n$, which are  large enough:
\begin{align}
\Pr{\text{pivot }< \mu - \delta n}\leq (2f(n)+1) \cdot{\alpha}^{f(n)} \leq \frac{1}{48}\epsilon .\label{eq:prob_small}
\end{align}
The analogous inequality holds for 
$\Pr{\text{pivot }> \mu + \delta n}$. Because $T_{\mathrm{pivot}}(n)\in o(n)$, we have
\begin{align}
T_{\mathrm{pivot}}(n)\leq \frac{1}{8}\epsilon n\label{eq:tpiv_small} . 
\end{align}
  for $n$ large enough. Now, we choose $n_0$ such that \prettyref{eq:prob_small} and \prettyref{eq:tpiv_small} hold  for $n\geq n_0$ and such that we have $(2+\epsilon) \delta + \frac{2}{n_0}< \frac{\epsilon}{4}$. We set $D = \Ep(n_0)+1$. Hence for $n < n_0$ the desired result \prettyref{eq:to_show} holds. Now, let $n\geq n_0$.
From \prettyref{eq:gunnar} we  obtain by symmetry:
\begin{align*}
 \Ep(n) &\leq n - 1 + T_{\mathrm{pivot}}(n) 
+ \sum_{k=\ceil{\mu - \delta n}}^{\floor{\mu+\delta n}} \Pr{\text{pivot }=k} \cdot \Ep(k-1)\\
&\qquad\qquad\qquad\qquad\hspace{2.4mm} + 2 \sum_{k=\floor{\mu+\delta n}+1}^{n}\Pr{\text{pivot }=k}\cdot \Ep(k-1) .\\
\intertext{Since $E$ is monotone, $E(k)$ can be bounded by the highest value in the respective interval:}
 &\leq n +\frac{1}{8}\epsilon n 
 + \Pr{\mu-\delta n\leq \text{pivot }\leq \mu+\delta n} \cdot\Ep\left(\floor{\mu+\delta n}\right)
 \\ &\qquad\qquad\hspace{3.18mm} +
2 \Pr{\text{pivot }> \mu+\delta n}
\cdot \Ep(n-1) 
\\
&\leq n  +\frac{1}{8}\epsilon n+ \left(1-\frac{1}{24}\epsilon\right) \cdot \Ep\left(\floor{\mu+\delta n}\right)  + 2 \frac{1}{48}\epsilon
\cdot \Ep(n-1) .\nonumber \\
\intertext{By induction we assume $E(k) \leq (2+\epsilon)k + D$ for $k<n$. Hence:}
\Ep(n) &\leq n  + \frac{1}{8}\epsilon n + \left(1-\frac{1}{24}\epsilon\right)\cdot \left((2+\epsilon)\cdot\left(\mu+\delta n\right)+D\right)
 + \frac{1}{24}\epsilon\cdot((2+\epsilon)n+D)\\
 &\leq n  +  (2 + \epsilon) \cdot \left(\frac{n+1}{2}+\delta n\right)
 + \frac{1}{8}\epsilon n + \frac{1}{24}\epsilon(2+\epsilon)n + D\\
& \leq 2n +1 +\frac{\epsilon}{2}+ (2 +\epsilon)\delta n + \frac{3}{4}  \epsilon n + D 
\quad < \quad (2+\epsilon)n + D .
\end{align*}
\end{proof}

\begin{corollary}[\cite{MartinezR01}]\label{cor:quickselect}
Let $f\in \omega(1)\cap o(n)$ with $1\leq f(n) \leq n$. When implementing Quickselect with the median of $f(n)$ randomly selected elements as pivot, the expected number of comparisons  is $2n + o(n)$.
\end{corollary}

\begin{proof}
In QuickHeapsort the recursion is always on the larger part of the array. Hence, the number of comparisons in partitioning for QuickHeapsort is an upper bound on the number of comparisons in Quickselect.
\end{proof}

In \cite{MartinezR01} it is also proved that choosing the pivot as median of $\mathcal{O}(\sqrt{n})$ elements is optimal for  Quicksort as well as for  Quickselect. This suggests that we choose the same value in QuickHeapsort; what is backed by our experiments.

\section{Modifications of QuickHeapsort Using Extra-space}\label{sec:modifications}
In this section we want to describe some modification of QuickHeapsort using $n$ bits of extra storage. 
We introduce two bit-arrays. In one of them (the CompareArray)~-- which is actually two bits per element~-- we store the comparisons already done (we need two bits, because there are three possible values~-- right, left, unknown~-- we have to store). 
 In the other one (the RedGreenArray) we store which element is red and which is green. 

Since the heaps have maximum size $n/2$, the RedGreenArray only requires $n/2$ bits. The CompareArray is only needed for the inner nodes of the heaps, i.e. length $n/4$ is sufficient. Totally this sums up to $n$ extra bits.

For the heap construction we do not use the algorithms described in \prettyref {sec:heap_con}. With the CompareArray we can do better by using the algorithm of McDiarmid and Reed \cite{McDiarmidR89}. 
The heap construction works similarly to Bottom-Up-Heapsort, i.e. the array is traversed backward calling for all inner positions $i$ the Reheap procedure on $i$. The Reheap procedure takes the subheap with root $i$ and restores the heap condition, if it is violated at the position $i$.
First, the Reheap procedure determines a \emph{special leaf} using the SpecialLeaf procedure as described in \prettyref{sec:QuickHeapsort}, but without moving the elements. Then, the final position of the former root is determined going upward from the special leaf (bottom-up-phase). In the end, the elements above this final position are moved up towards the root by one position. 
That means that all but one element which are compared during the bottom-up-phase, stay in their places. Since in the SpecialLeaf procedure these elements have been compared with their siblings, these comparisons can be stored in the CompareArray and can be used later.

 With another improvement concerning the construction of heaps with seven elements as in \cite{CarlssonCM94} the benefits of this array can be exploited even more.

The RedGreenArray is used during the sorting phase, only. Its functionality is straightforward: Every time a red element is inserted into the heap, the corresponding bit is set to red. The SpecialLeaf procedure can stop as soon as it reaches an element without green children. Whenever a red and a green element have to be compared, the comparison can be skipped.

\begin{theorem}\label{thm:exptimeModified}
Let $f\in \omega(1)\cap o(n)$ with $1\leq f(n) \leq n$, e.g., 
$f(n) = \lg n$, and 
let $\ET{n}$ be the expected number of comparisons
by QuickHeapsort using the CompareArray with the improvement of  \cite{CarlssonCM94} and the RedGreenArray on a fixed  input array of size $n$. 
Choosing the pivot as median of $f(n)$ randomly selected elements 
in time $\Oh(f(n))$, we have
\begin{align*}
\ET{n}&\leq n\lg n -0.997 n+ o (n) .
\end{align*}
\end{theorem}

\begin{proof}
We can analyze the savings by the two arrays separately, because the CompareArray only affects comparisons between two green elements, while the RedGreenArray only affects comparisons involving at least one red element. 

First, we consider the heap construction using the CompareArray. With this array we obtain the same worst case bound as for the standard heap construction method. However, the CompareArray has the advantage that at the end of the heap construction many comparisons are stored in the array and can be reused for the extraction phase. More precisely: For every comparison except the first one made when going upward from the special leaf, one comparison is stored in the CompareArray, since for every additional comparison one element on the path defined by SpecialLeaf stays at its place. Because every pair of siblings has to be compared at one point during the heap construction or extraction, all these stored comparisons can be reused. Hence, we only have to count the comparisons in the SpecialLeaf procedure during the construction plus $\frac{n}{2}$ for the first comparison when going upward. Thus, we get an amortized bound for the comparisons during construction of $\frac{3n}{2}$. 

In \cite{CarlssonCM94} the notion of \emph{Fine-Heaps} is introduced. A Fine Heap is a heap with the additional CompareArray such that for every node the larger child is stored in the array. Such a Fine-Heap of size $m$ can be constructed using the above method with $2m$ comparisons. In \cite{CarlssonCM94} Carlsson, Chen and Mattsson showed that a Fine-Heap of size $m$ actually can be constructed with only $\frac{23}{12}m+\mathcal{O}(\lg^2m)$ comparisons. That means we have to invest  $\frac{23}{12}m+\mathcal{O}(\lg^2m)$ for the heap construction and at the end there are $\frac{m}{2}$ comparisons stored in the array. All these comparisons stored in the array are used later. Summing up over all heaps during an execution of QuickHeapsort, we can save another $\frac{1}{12}n$ comparisons additionally to the comparisons saved by the CompareArray with the result of \cite{CarlssonCM94}.
Hence, for the amortized cost of the heap construction $T_{\mathrm{con}}^{\mathrm{amort}}$ (i.e. the number of comparisons needed to build the heap minus the number of comparisons stored in the CompareArray after the construction which all can be reused later) we have obtained:
\begin{proposition}\label{prop:t_con_amort}
$T_{\mathrm{con}}^{\mathrm{amort}}(n) \leq  \frac{17}{12} n + o(n)$.
\end{proposition}
This bound is slightly better than the average case for the heap construction with the algorithm of \cite{McDiarmidR89} which is $1.52n$.

Now, we want to count the number of comparisons we save using the RedGreenArray. We distinguish the two cases that two red elements are compared and that a red and a green element are compared. 
Every position in the heap has to turn red at one point. At that time, all nodes below this position are already red. Hence, for that element we save as many comparisons as the element is above the bottom level. Summing over all levels of a heap of size $m$ the saving results in 
$ \approx \frac{m}{4} \cdot 1 + \frac{m}{8} \cdot 2 + \cdots = m \cdot \sum\limits_{i\geq 1}i2^{-i-1}= m .$
This estimate is exact up to $\Oh(\lg m)$-terms. Since the expected number of heaps is $\Oh(\lg n)$, we obtain for the overall saving 
the value 
$T_{\mathrm{saveRR}}(n) = n + \Oh(\lg^2 n).$

Another place where we save comparisons with the RedGreenArray is when a red element is compared with a green element.
 It occurs at least one time~-- when the node looses its last green child~-- for every inner node that we compare a red child with a green child. Hence, we save at least as many comparisons as there are inner nodes with two children, i.e. at least $\frac{m}{2} - 1$.
Since every element~-- except the expected $\Oh(\lg n)$ pivot elements~-- is part of a heap exactly once, we save at least
$T_{\mathrm{saveRG}}(n) \geq \frac{n}{2} + \Oh(\lg n)$
comparisons when comparing  green with  red elements. In the average case the saving might be even slightly higher, since comparisons can also be saved when a node does not loose its last green child.

Summing up all our savings and using the median of $f(n)\in \omega(1)\cap o(n)$  as pivot we obtain the proof of \prettyref{thm:exptimeModified}:
\begin{align*}
\ET{n} &\leq T_{\mathrm{con}}^{\mathrm{amort}}(n)
+T_{\mathrm{ext}}(n) +{\mathbb{E}}[T_{\mathrm{part}}(n)] -T_{\mathrm{saveRR}}(n)-T_{\mathrm{saveRG}}(n)\\
	&\leq  \frac{17}{12}n + n\cdot(\lfloor\lg n\rfloor-3) + 2 \os n +  2n -\frac{3n}{2}+ o(n)\\
	&\leq n\lg n -0.997 n+ o (n) .
\end{align*}
\end{proof}
\section{Experimental Results and Conclusion}\label{sec:experiments}



 In \prettyref{fig:different_QHS_versions} we present the number of comparisons of the different versions of QuickHeapsort we considered in this paper, i.e. the basic version, the improved variant of \prettyref{sec:QuickHeapsort}, and the version using bit-arrays (however, without the modification by \cite{CarlssonCM94}) for different values of $n$. We compare them with Quicksort, Ultimate Heapsort, Bottom-Up-Heapsort and MDR-Heapsort.  All algorithms are implemented with median of $\sqrt{n}$ elements as pivot (for Quicksort we show additionally the data with median of 3).
For the heap construction we implemented the normal algorithm due to Floyd \cite{Floyd64} as well as the algorithm using the extra bit-array (which is the same as in MDR-Heapsort). 

\begin{figure}[ht]
\begin{center}
\begin{tikzpicture}
	\begin{axis}[domain=800 : 4161888,xmode= log, xlabel={$n$},ylabel={$(\text{\#comparisons} - n\lg n)/n$},width=0.8\textwidth, legend pos= north east ]
 		\addplot[densely dashed,every mark/.append style={solid,fill=black},mark=square*]  table[x=n,y=QS3] {vergleiche5.csv};
		\addlegendentry{\scriptsize{ Quicksort with Median of 3}}
            	\addplot[mark=diamond*,every mark/.append style={solid,fill=black}]  table[x=n,y=QSS] {vergleiche5.csv};
		\addlegendentry{\scriptsize{Quicksort with Median of $\sqrt{n}$}}
		\addplot[mark=triangle*,every mark/.append style={solid,fill=black}, densely dashed]  table[x=n,y=bQHSS] {vergleiche5.csv};
		\addlegendentry{\scriptsize{Basic QuickHeapsort}}
		\addplot[every mark/.append style={solid,fill=black},mark=*]  table[x=n,y=iQHSS] {vergleiche5.csv};
		\addlegendentry{\scriptsize{Improved QuickHeapsort}}
		\addplot[ mark=triangle*,every mark/.append style={solid,fill=black}]  table[x=n,y=QHSbaS] {vergleiche5.csv};
		\addlegendentry{\scriptsize{QuickHeapsort with bit-arrays}}
		\addplot[densely dashed,every mark/.append style={solid,fill=black} ,mark=*]  table[x=n,y=MDR] {vergleiche5.csv};
		\addlegendentry{\scriptsize{MDR-Heapsort}}
		\addplot[mark=triangle*,every mark/.append style={solid,fill=black}]  table[x=n,y=Ultimate] {vergleiche5.csv};
		\addlegendentry{\scriptsize{Ultimate-Heapsort}}
 		\addplot[dashed]  {- 1.44};
 		\addlegendentry{\scriptsize{Lower Bound}}
       \end{axis}
\end{tikzpicture} 
\end{center}
\caption{Average number of comparisons of QuickHeapsort implemented with median of $\sqrt{n}$ compared with other algorithms}\label{fig:different_QHS_versions}
\end{figure}
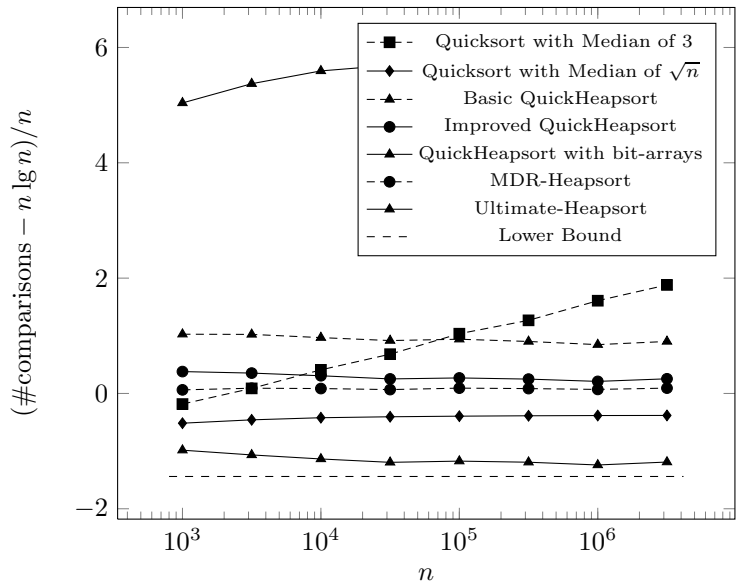

More results with other pivot selection strategies are in \prettyref{tab:different_QHS_versions} and \prettyref{tab:pivot_selection} in \prettyref{app:tables} confirming that a sample size of $\sqrt{n}$ is optimal for pivot selection with respect to the number of comparisons and also that the $o(n)$-terms in  \prettyref{thm:exptime} and \prettyref{thm:otto} are not too big.
In \prettyref{tab:running_times} in \prettyref{app:tables} we present actual running times of the different algorithms for $n=1000000$. 
 All the numbers, except the running times, are average values over 100 runs with random data. As our theoretical estimates predict, QuickHeapsort with bit-arrays beats all other variants including Relaxed-Weak-Heapsort (see \prettyref{tab:different_QHS_versions},  \prettyref{app:tables}) when implemented with median of $\sqrt{n}$ for pivot selection.
It also performs $326728 \approx 0.33 \cdot 10^6$ comparisons less than our theoretical predictions which are $10^6\cdot\lg (10^6 )- 0.9139 \cdot 10^6 \approx 19017569$ comparisons.

In this paper we have shown that with known techniques QuickHeapsort can be implemented with expected number of comparisons less than $n\lg n - 0.03n + o(n)$ and extra storage $O(1)$. On the other hand, using $n$ extra bits we can improve this to $n\lg n - 0.997n + o(n)$, i.e. we showed that QuickHeapsort can compete with the most advanced Heapsort variants. These theoretical estimates were also confirmed by our experiments. 
We also considered different pivot selection schemes. 
For any constant size sample for pivot selection, QuickHeapsort beats Quicksort for large $n$, since Quicksort has a expected running time of $\approx Cn\lg n$ with $C>1$. However, when choosing the pivot as median of $\sqrt{n}$ elements (i.e. with the optimal strategy) then our experiments show that  Quicksort needs less comparisons than QuickHeapsort. However,
using bit-arrays QuickHeapsort is the winner, again.
In order to make the last statement rigorous, better theoretical 
bounds for Quicksort with sampling $\sqrt{n}$ elements are needed.
For future work it would also be of interest to prove the optimality of $\sqrt{n}$ elements for pivot selection in QuickHeapsort, to estimate the lower order terms of the average running time of QuickHeapsort and also to find an exact average case analysis for the saving by the bit-arrays.

\subsubsection*{Acknowledgements.}
We thank Martin Dietzfelbinger, Stefan Edelkamp and Jyrki Katajainen for their helpful comments.
We thank Simon Paridon for implementing the algorithms for our experiments.

\newcommand{\Ju}{Ju}\newcommand{\Ph}{Ph}\newcommand{\Th}{Th}\newcommand{\Ch}{Ch}\newcommand{\Yu}{Yu}\newcommand{\Zh}{Zh}

\newpage
\section*{APPENDIX}
\appendix
\section{Proofs}\label{app:proof}

\begin{proof}[Proof of \prettyref{lm:F_ineq}]
Since the right derivative is monotonically increasing we have:
\begin{align*}
F(x+\delta)-F(x) = \int_x^{x+\delta} F'(t) \,\mathrm dt \geq F'(x)\cdot \delta = (\floor{\lg x}+2)\delta
\end{align*}
and
\begin{align*}
F(y)-F(y-\delta)= \int_{y-\delta}^y F'(t) \,\mathrm dt \leq F'(y)\cdot \delta = (\floor{\lg y}+2)\delta .
\end{align*}
This yields:
\begin{align*}
F(y)-F(y-\delta) &\leq  (\floor{\lg y}+2)\delta
	\leq  (\floor{\lg x}+2)\delta
	\leq F(x+\delta)-F(x) .
\end{align*}
By adding $F(x)+ F(y-\delta)$ on both sides we obtain the first claim of \prettyref{lm:F_ineq}. 
Note that $\lim_{\eps \to 0}F(\eps) =0$. Hence 
the second claim follows from the first by considering the limit $\delta \to y$. 
\end{proof}

\begin{proof}[Proof of \prettyref{lm:prob_bound}.]
First note that the probability for choosing the $k$-th element as pivot satisfies 
\[\binom{n}{2c+1}\cdot \Pr{\text{pivot }= k } =  \binom{k-1}{c}\binom{n-k}{c}{} .\]
We use the notation of \emph{falling factorial} $x^{\underline \ell}
= x \cdots (x-\ell +1)$. Thus, $\binom x \ell = \frac{x^{\underline \ell}}{\ell!}$. 

\begin{align*}
\Pr{\text{pivot }= k } &= 
\dfrac{(2c+1)!\cdot(k-1)^{\underline c}\cdot(n-k)^{\underline c}}{(c!)^2 \cdot n^{\underline {2c +1}}}\\
&= \binom{2c}{c}(2c+1) \frac{1}{(n-2c)} \prod_{i=0}^{c-1}\frac{(k-1-i)(n-k-i)}{(n-2i-1) (n-2i)} .
\end{align*}
For  $k \leq c$ we have  $\Pr{\text{pivot }= k } =0$. So, let  $c < k \leq  \frac{n}{2} - \delta n$ and let us consider an index $i$ in the product with  $0\leq i< c$.

\begin{align*}
\frac{(k-1-i)(n-k-i)}{(n-2i-1) (n-2i)} &\leq \frac{(k-i)(n-k-i)}{(n-2i) (n-2i)}\\
  &= \frac{\left(\left(\frac{n}{2}-i\right)-\left(\frac{n}{2}-k\right)\right) \cdot \left(\left(\frac{n}{2}-i\right)+\left(\frac{n}{2}-k\right)\right) }{\left(n-2i\right)^2}\\
&= \frac{\left(\frac{n}{2}-i\right)^2-\left(\frac{n}{2}-k\right)^2}{\left(n-2i\right)^2}\\
&\leq \frac{1}{4}- \frac{\left(\frac{n}{2}-\left(\frac{n}{2}-\delta n\right)\right)^2}{n^2}  
= \frac{1}{4}-\delta^2 .
\end{align*}
We have  $\binom{2c}{c}\leq 4^c$. Since $2c +1 \leq\frac{n}{2}$, we obtain: 
\begin{align*}
\Pr{\text{pivot }= k } 
&\leq 4^c(2c+1) \frac{1}{(n-2c)}\left(\frac{1}{4}-\delta^2\right)^c
< (2c+1) \frac{2}{n}\alpha^c .
\end{align*}
Now, we obtain the desired result. 
\begin{align*}
\Pr{\text{pivot }\leq \frac{n}{2} - \delta n}
&< \sum_{k=0}^{\floor{\frac{n}{2} - \delta n}} (2c+1) \frac{2}{n}\alpha^c 
\,\leq\, (2c+1) \alpha^c 
\end{align*}
\end{proof}

\section{More Experimental Results}\label{app:tables}

In \prettyref{tab:running_times} we present actual running times of the different algorithms for $n=1000000$ with two different comparison functions (the numbers displayed here are averages over 10 runs with random data). 
 One of them is the normal integer comparison, the other one first applies four times the logarithm to both operands before comparing them. Like in \cite{EdelkampS02}, this simulates expensive comparisons.

 In \prettyref{tab:different_QHS_versions} all algorithms are implemented with median of $3$ and with median of $\sqrt{n}$ elements as pivot. %
 We compare them with Quicksort implemented with the same pivot selection strategies, Ultimate Heapsort, Bottom-Up-Heapsort and MDR-Heapsort. In \prettyref{tab:different_QHS_versions} we also added the values for Relaxed-Weak-Heapsort which were presented in \cite{EdelkampS02}.

\begin{table}[ht]
\caption{Running times for QuickHeapsort and other algorithms tested on $10^6$, average over 10 runs elements}\label{tab:running_times}
\begin{center}
\begin{tabular}{|l|l|l|}\hline
{\bf Sorting algorithm} &\parbox{18.5mm} {integer data time $[s]$}& \parbox{27.5mm} {$\log^{(4)}$-test-function time $[s]$} \\\hline
Basic QuickHeapsort, median of 3 & ~ 0.1154 &~ 4.21 \\\hline
Basic QuickHeapsort, median of $\sqrt n$ &~ 0.1171 &~ 4.109\\\hline
Improved QHS, median of 3 &~ 0.1073 &~ 4.049\\\hline
Improved QHS, median of $\sqrt n$ &~ 0.1118 &~ 3.911 \\\hline
QHS with bit-arrays, median of $3$ &~ 0.1581 &~ 3.756 \\\hline
QHS with bit-arrays, median of $\sqrt n$ &~ 0.164 &~ 3.7 \\\hline
Quicksort with median of 3  &~ 0.1181 &~ 3.946\\\hline
Quicksort with median of $\sqrt n$ &~ 0.1316 &~ 3.648\\\hline
 Ultimate Heapsort &~ 0.135 &~ 5.109\\\hline
 Bottom-Up-Heapsort &~ 0.1677 &~ 4.132\\\hline
 MDR-Heapsort &~ 0.2596 &~ 4.129\\\hline
\end{tabular}
\end{center}
\end{table}

\begin{table}[ht]
\caption{QuickHeapsort and other algorithms tested on $10^6$ elements  (the data for Relaxed-Weak-Heapsort is taken from \cite{EdelkampS02}). }\label{tab:different_QHS_versions}
\begin{center}
\begin{tabular}{|l|l|}\hline
{\bf Sorting algorithm} &\parbox{40.5mm}{Average number of comparisons for $n=10^6$}\\\hline 
Basic QuickHeapsort with median of 3 & 21327478\\\hline
Basic QuickHeapsort with median of $\sqrt n$ & 20783631\\\hline
Improved QuickHeapsort, median of 3 & 20639046\\\hline
Improved QuickHeapsort, median of $\sqrt n$ & 20135688\\\hline
QuickHeapsort with bit-arrays, median of $3$ & 19207289\\\hline
QuickHeapsort with bit-arrays, median of $\sqrt n$ & 18690841 $\qquad *\text{Best result}* $\\\hline
Quicksort with median of 3 & 21491310\\\hline
Quicksort with median of $\sqrt n$ & 19548149\\\hline
Bottom-Up-Heapsort & 20294866\\\hline
MDR-Heapsort & 20001084\\\hline
Relaxed-Weak-Heapsort & 18951425\\\hline\hline
{\bf Lower Bound:} $\lg{n!}$ & 18488884 $ \approx \lg{(10^{6}!)}$\\\hline
\end{tabular}
\end{center}
\end{table}

We also compare the different pivot selection strategies on the basic QuickHeapsort with no modifications. 
 We test sample of sizes of  one, three, approximately $\lg n$, $\sqrt[4]{n}$,$\sqrt{n/\lg n}$,  $\sqrt{n}$, and $ n^{\frac{3}{4}}$ for the pivot selection. 

 In \prettyref{tab:pivot_selection} the average number of comparisons and the standard deviations 
are listed. We ran the algorithms on arrays of length 10000 and one million. The displayed data is the average resp.\ standard deviation of 100 runs of QuickHeapsort with the respective pivot selection strategy.

These results are not very surprising: The larger the samples get, the smaller is the standard deviation. The average number of comparisons reaches its minimum with a sample size of approximately $\sqrt{n}$ elements. One notices that the difference for the average number of comparisons is relatively small, especially between the different pivot selection strategies with non-constant sample sizes. This confirms experimentally that the $o(n)$-terms in  \prettyref{thm:exptime} and \prettyref{thm:otto} are not too big.
\begin{table}[hbt]
\caption{Different strategies for pivot selection for basic QuickHeapsort tested on $10^4$ and $10^6$ elements. The standard deviation of our experiments is given in percent of the average number of comparisons. }\label{tab:pivot_selection}
{%
\newcommand{\mc}[3]{\multicolumn{#1}{#2}{#3}}
\begin{center}
\begin{tabular}{|l|l|l|l|l|}\hline
$n$ & \mc{2}{l|}{$10^4$} & \mc{2}{l|}{$10^6$}\\\hline
Sample size &\parbox{26mm}{Average number of comparisons} &\parbox{15mm}{Standard deviation}   &\parbox{26mm}{Average number of comparisons} &\parbox{15mm}{Standard deviation}\\\hline
$1$ & 152573 & 4.281 & 21975912 & 3.452\\\hline
$3$ & 146485 & 2.169 & 21327478 & 1.494\\\hline
$\sim \lg n$ & 143669 & 0.954 & 20945889 & 0.525\\\hline
$\sim \sqrt[4]{n}$ & 143620 & 0.857 & 20880430 & 0.352\\\hline
$\sim \sqrt{n/\lg n}$ & 142634 & 0.413 & 20795986 & 0.315\\\hline
$\sim \sqrt{n}$ & 142642 & 0.305 & 20783631 & 0.281\\\hline
$\sim n^{\frac{3}{4}}$ & 147134 & 0.195 & 20914822 & 0.168\\\hline
\end{tabular}
\end{center}
}
\end{table}

\section{Some Words about the Worst Case Running Time} 

Obviously the worst case running time depends on how the pivot element is chosen. If just one random element is used as pivot we get the same quadratic worst case running time as for Quicksort. However the probability that in QuickHeapsort we run in such a ``bad case'' is not higher than in Quicksort, since any choice of pivot elements leading to a worst case scenario in QuickHeapsort also yields the worst case for Quicksort. 

If we choose the pivot element as median of approximately $2\lg n$ elements, we get a worst case running time of $\mathcal{O}\left(\frac{n^2}{\lg n}\right)$,
i.e. for the worst case it makes almost no difference, if the pivot is selected as median of $2\lg n$ or just as one random element.

However, if we use  approximately $\frac{n}{\lg n}$ elements as sample for the pivot selection, we can get a better bound on the worst case.

Let $f:\N\to \N_{\geq 1}$ be some monotonically growing function with $f\in o(n)$
 (e.g.\ $f(n) = \lg n$).  We can apply the ideas of the Median of Medians algorithm \cite{BlumFPRT73}: First we choose  $\frac{n}{f(n)}$ random elements, then we group them into groups of five elements each. The median of each group can be determined with six comparisons \cite[p. 215]{Knuth}. Now, the median of these medians can be computed using Quickselect. We assume that Quickselect is implemented with the same strategy for pivot selection. That means we get the same recurrence relations for the worst case complexity of the partitioning-phases in QuickHeapsort and for the worst case of Quickselect: 
\begin{align*}
T(n) 	&= n + \frac{6n}{5f(n)} + T\left(\frac{n}{5f(n)}\right) + T\left(n-\frac{3n}{10f(n)}\right) .
\end{align*}
This yields $T(n)\leq c nf(n)$ for some $c$ large enough.
Hence with this pivot selection strategy, we reach a worst case running time for QuickHeapsort of $n\lg n + \mathcal{O}(nf(n))$ and~-- if $f(n) \in \omega(1)$~-- average running time as stated in \prettyref{sec:analysis}. 

Driving this strategy to the end and choosing $f(n)=1$ leads to Ultimate Heapsort (or better a slight modification of it~-- and Quickselect turns into the Median of Medians algorithm). Then we have $T(n) = n\lg n + \mathcal{O}(n)$ for the worst case of QuickHeapsort. However, our bound for the average case does not hold anymore.

In order to obtain an $n\lg n + \mathcal{O}(n)$-bound for the worst case without loosing our bound for the average case, we can apply a simple trick: 
Whenever after the partitioning it turns out that the pivot does not lie in the interval $\{\frac{n}{4} , \dots , \frac{3n}{4} \}$ we switch to Ultimate Heapsort.
This immediately yields the worst case bound of  $n\lg n + \mathcal{O}(n)$. 
Moreover, the proof of \prettyref{thm:otto} can easily be changed in order to deal with this modification: Let $C\cdot n$ be the worst case number of comparisons for pivot selection and partitioning in Ultimate Heapsort. We can change \prettyref{eq:prob_small} to
\[\Pr{\text{pivot }< \mu - \delta n} \leq \frac{1}{8C}\epsilon .\]
Then, the rest of the proof is exactly the same.
 Hence, \prettyref{thm:otto} and \prettyref{thm:exptime} are also valid when switching to Ultimate Heapsort in the case of a `bad' choice of the pivot.

\section{Pseudocode of Basic QuickHeapsort}\label{app:pseudo_basic}
\begin{algorithmus}{}\label{alg:quick_heapsort}%
  \tab \coprocedure QuickHeapsort($A[1 .. n]$)
  \\ \tab \cobegin  
\\ \tabb \coif $n>1$ \cothen
   	\\ \tabbb $p:=$ ChoosePivot;
   	\\ \tabbb $k :=$ PartitionReverse($A[1.. n]$, $p$); 
   	\\ \tabbb  \coif $k \leq n / 2$    \cothen
		\\ \tabbbb TwoLayerMaxHeap($A[1 .. n]$, $k-1$);\comment{heap-area: $\{1.. k - 1\}$} 
 		\\ \tabbbb swap($A[k],A[n-k+1])$;
		\\ \tabbbb QuickHeapsort($A[1.. n-k]$); \comment{recursion}
	\\ \tabbb  \coelse
 		\\ \tabbbb TwoLayerMinHeap($A[1 .. n]$, $n-k$); \comment{heap-area: $ \{k + 1.. n\}$}
 		\\ \tabbbb swap($A[k],A[n-k+1])$;
 		\\ \tabbbb QuickHeapsort($A[(n-k+2) .. n]$); \comment{recursion}
	\\ \tabbb  \coendif
  \\ \tabb  \coendif	
  \\ \tab \coendprocedure 
 \end{algorithmus} 
The ChoosePivot function returns an element $p$ of the array chosen as pivot. 
 The PartitionReverse function returns an index $k$ and rearranges the array $A$ so that $p =  A[k]$, $A[i]\geq A[k]$ for $i< k$ and $A[i]\leq A[k]$ for $i> k$ using $n-1$ comparisons.
\begin{algorithmus}{}\label{alg:special_leaf}%
  \cofunction SpecialLeaf$(A[1 .. m])$:
  \\ \tab \cobegin
  \\ \tabb  $i:=1$;
  \\ \tabb  \cowhile $2i\leq m$ \codo \comment{i.e. while $i$ is not a leaf}
  \\ \tabbb  \coif $2i+1 \leq m$ \coand $A[2i+1] > A[2i]$ \cothen 
  \\ \tabbbb   $A[i]:=A[2i+1]$;
  \\ \tabbbb   $i:=2i+1$;
  \\ \tabbb  \coelse
  \\ \tabbbb   $A[i]:=A[2i]$;
  \\ \tabbbb   $i:=2i$;
  \\ \tabbb \coendif
  \\ \tabb \coendwhile
  \\ \tabb \coreturn $i$;
  \\ \tab \coendfunction
\end{algorithmus}

\begin{algorithmus}{}\label{alg:max_heap}%
  \tab \coprocedure TwoLayerMaxHeap($A[1 .. n]$, $m$)
  \\ \tab \cobegin
 	 \\ \tabb  ConstructHeap$(A[1 .. m])$;
 	 \\ \tabb  \cofor $i := 1 $ \coto $ m $ \codo
 	 \\ \tabbb $temp :=  A[n - i + 1]$;
 	 \\ \tabbb $A[n - i + 1] := A[1]$;
	 \\ \tabbb $j:=$SpecialLeaf$(A[1 .. m])$;
	 \\ \tabbb $A[j]:=temp$;
	 \\ \tabb \coendfor
  \\ \tab \coendprocedure
 \end{algorithmus} 
The procedure TwoLayerMinHeap is symmetric to TwoLayerMaxHeap, so we do not present its pseudocode here.

\end{document}